\documentclass[conference]{IEEEtran}

\usepackage{amsmath,amssymb,amsthm}
\usepackage{graphicx}
\usepackage{booktabs}
\usepackage{array}
\usepackage{microtype}
\usepackage{placeins}
\usepackage{balance}
\usepackage{xurl}
\usepackage[hidelinks]{hyperref}
\usepackage{orcidlink}
\usepackage{xcolor}

\newif\ifanonymous
\anonymousfalse 

\newcommand{\RC}{\textsc{ReputationChain}}
\newcommand{\clip}{\operatorname{clip}_{[0,1]}}
\newcommand{\R}{\mathit{Rs}}
\newcolumntype{L}[1]{>{\raggedright\arraybackslash}p{#1}}

\theoremstyle{plain}
\newtheorem{proposition}{Proposition}

\title{ReputationChain: Robust Trust Updating for\\Blockchain-Enabled Supply Chains}

\hypersetup{
  pdftitle={ReputationChain: Robust Trust Updating for Blockchain-Enabled Supply Chains},
  pdfsubject={Blockchain-enabled supply-chain reputation and trust management},
  pdfkeywords={blockchain, supply chain, reputation, trust management, collusion, identity multiplicity, sparse-history fairness}
}

\ifanonymous
  \hypersetup{pdfauthor={Anonymous Author(s)}}
  \author{\IEEEauthorblockN{Anonymous Author(s)}}
\else
  \hypersetup{pdfauthor={Adnan Iftekhar, Chengliang Zheng, Xiaohui Cui, Mir Hassan}}
  \author{
  \IEEEauthorblockN{
  Adnan Iftekhar\,\orcidlink{0000-0002-8249-6876}\textsuperscript{a,e},
  Chengliang Zheng\textsuperscript{b},
  Xiaohui Cui\textsuperscript{a,*},
  and Mir Hassan\textsuperscript{c,d}}
  \IEEEauthorblockA{\scriptsize
  \textsuperscript{a}\textit{Key Laboratory of Aerospace Information Security and Trusted Computing, Ministry of Education,}\\[-0.2ex]
  \textit{School of Cyber Science and Engineering, Wuhan University, Wuhan 430072, China}\\[-0.2ex]
  \textsuperscript{b}\textit{School of Artificial Intelligence (School of Blockchain Industry), Chengdu University of Information Technology, Chengdu 610225, China}\\[-0.2ex]
  \textsuperscript{c}\textit{Human Environment Technology Systems Centre, Mykolas Romeris University, Vilnius LT-08320, Lithuania}\\[-0.2ex]
  \textsuperscript{d}\textit{Azerbaijan State Oil and Industry University, 34 Azadliq Avenue, Baku, Azerbaijan}\\[-0.2ex]
  \textsuperscript{e}\textit{The University of Faisalabad, Faisalabad, Pakistan}\\[-0.2ex]
  Emails: adnan@whu.edu.cn; zhengcl@cuit.edu.cn; xcui@whu.edu.cn; mir.hassan@mruni.eu\\[-0.2ex]
  \textsuperscript{*}Corresponding author: Xiaohui Cui (xcui@whu.edu.cn).}
  }
\fi

\begin{document}
\maketitle

\begin{abstract}
A blockchain can preserve the record of a supply-chain interaction, but an intact ledger says nothing about whether the participant behind that record deserves trust in the next risk-sensitive transaction. Existing reputation designs concentrate on product evidence, global feedback aggregation, or review authenticity, and three problems that dominate permissioned supply chains receive far less attention: inflation through repeated bilateral interactions, the advantage gained from identity multiplicity, and decay rules that quietly punish honest participants whose history is simply sparse. We present \RC, a participant trust framework that treats the blockchain as an evidence and provenance layer rather than as the source of trust itself. Governed interaction outcomes become bounded evidence; repeated interactions between the same pair are discounted, low counterparty diversity is penalized, governance-supplied identity confidence weights positive evidence, and scores decay toward a neutral prior at a rate set by verified interaction volume. Identity, contract, outcome, and update provenance stay on chain, while the nonlinear computation runs off chain and is checked on chain for admissibility. In controlled simulations over 30 seeded runs with matched interaction traces, the full model cuts mean collusive gain to 0.1443 (against 0.3688 for naive mean evidence and 0.3585 for static decay). With ten identities under one controller, the reputation inflation ratio drops to 0.8723 while three comparison baselines stay above 1.08. On identical newcomer traces, volume-aware decay raises mean newcomer reputation from 0.6626 to 0.7589 and lowers the false low-trust rate from 0.3633 to 0.1683; seed-paired analyses confirm that these differences hold run by run. What the evidence supports is a bounded reduction in reputation distortion, not attacker detection, and it is simulation evidence: deployment measurement and calibration against operational data remain necessary before production use.
\end{abstract}


\section{Introduction}
Blockchain platforms now serve as shared recordkeeping infrastructure across many supply-chain consortia, and the case for them usually rests on traceability and auditability \cite{kumar2025blockchainscm,karaduman2025blockchainreview}. Those properties settle disputes over whether an event was recorded and whether the record was later altered. They leave a different question open. Should this supplier, this carrier, or this retailer be trusted in the next transaction that creates real financial or operational exposure? A ledger can hold a perfectly preserved record of a delayed settlement or a failed delivery and still offer no defensible way of turning that evidence into a participant-level trust signal.

The question is hardest in inclusive supply chains, where small firms, cross-border partners, temporary collaborators, and third-party providers arrive with little shared history \cite{desouzamiguel2026supplierdiversity,medina2024blockchainagency}. A trust mechanism for this setting has to steer between two failure modes at once. If repeated successful interactions inside a narrow group are allowed to compound without limit, reputation becomes something a closed clique can manufacture. If thin history is read as evidence of unreliability, the mechanism ends up protecting incumbents rather than the supply chain, and it locks newcomers out of exactly the interactions that would let them build a record.

Prior work supplies important pieces of the answer. TrustChain fuses multiple evidence sources for food supply chains \cite{malik2019trustchain}, and DeTRM models commodity and participant trust together in blockchain-supported supply chains \cite{putra2022detrm}. MeritRank bounds the benefit of Sybil identities in graph-based reputation \cite{nasrulin2022meritrank}, while BlockRep ties review authenticity to cryptographic purchasing evidence \cite{zhao2025blockrep}. Each of these solves an adjacent problem. None of them jointly targets participant-level trust updating under governed interactions, where repeated bilateral reinforcement, governance-side identity assurance, and sparse history all feed into later risk-sensitive decisions.

We designed \RC\ with that gap in mind. The framework uses the blockchain as an evidence and state-management layer, not as the source of trust. Registered participants interact under governing contracts; verified completion, delay, failure, and dispute events become bounded evidence; and the reputation layer then applies repeated-pair discounting, counterparty-diversity adjustment, identity-confidence weighting, and volume-aware decay toward a neutral prior.

This paper makes four contributions.
\begin{enumerate}
  \item We define an auditable participant trust architecture that separates blockchain-recorded evidence from the numerical reputation computation.
  \item We develop a bounded, event-indexed update model that limits the influence of repeated bilateral interactions and of identity multiplicity while preserving adverse evidence in full.
  \item We introduce volume-aware decay toward a neutral prior, which lowers the rate at which honest newcomers are falsely classified as low trust.
  \item We evaluate the mechanisms against baselines, through targeted ablations, and under parameter sensitivity sweeps, with 30 seeded runs and matched traces per condition, and we report seed-paired effects with confidence intervals.
\end{enumerate}

The claims are deliberately narrow. \RC\ does not detect every colluder, and it does not identify Sybil identities; what it does, under the tested assumptions, is reduce how much reputation an attacker can extract. The distinction is worth stating up front, because our strongest identity-multiplicity experiment still leaves a high false trust elevation rate, and we report that number rather than bury it.

\section{Related Work and Research Gap}
\subsection{Reputation in Distributed Systems}
EigenTrust, PeerTrust, and the Beta Reputation System set the main directions for distributed reputation aggregation \cite{kamvar2003eigentrust,xiong2004peertrust,josang2002beta}: transitive global trust, feedback weighted by transaction context, and probabilistic estimates derived from binary outcomes. All three grew out of open peer communities, and their assumptions show it. In the setting we target, identities are registered, interactions are governed by contracts, and a reputation score feeds directly into a participant's future business exposure.

The attack surface of reputation systems is well catalogued: collusion, reciprocal inflation, badmouthing, whitewashing, and identity manipulation \cite{hoffman2009reputationattacks}, with the Sybil attack as the canonical demonstration that cheap identity creation can distort distributed decisions \cite{douceur2002sybil}. Permissioned membership raises the cost of entry but does not close the identity problem. A participant can still control several valid identities through subsidiaries or affiliated accounts, and governance processes themselves can hand out identity assurance of uneven quality.

\subsection{Blockchain-Supported Supply-Chain Trust}
TrustChain combines sensor, buyer, and regulator ratings under time decay \cite{malik2019trustchain}; DeTRM links commodity evidence, agreement fulfilment, authority endorsement, and trust propagation \cite{putra2022detrm}; and a broader line of work integrates blockchain with supply-chain trust management and smart manufacturing \cite{wu2022scmtrust}. Together, these systems establish that a blockchain can preserve evidence and automate trust-related state transitions.

Product evidence and participant evidence are not the same thing, though. A sensor stream can support claims about the condition of goods without saying anything about how reliably a counterparty behaves across partners, and a participant can perform flawlessly for one closely connected partner while remaining untested everywhere else. \RC\ therefore attends to the shape of the interaction evidence, not just its existence.

MeritRank offers a stronger formal treatment of Sybil-tolerant graph reputation \cite{nasrulin2022meritrank}, and BlockRep offers stronger cryptographic protection for review authenticity and anonymity \cite{zhao2025blockrep}. We do not attempt to replace either. The gap we address is narrower and operational: bounded participant trust updates inside a governed supply-chain consortium. Table~\ref{tab:positioning} summarizes where each system sits relative to ours.

One naming clarification is needed. RepChain \cite{huang2021repchain} uses validator reputation for shard formation, leader selection, and incentives, and therefore operates inside the consensus layer. \RC\ operates above it: the scores describe business participants and never touch validator selection or consensus voting.

\begin{table*}[t]
\centering
\caption{Positioning of \RC\ relative to representative reputation systems.}
\label{tab:positioning}
\scriptsize
\setlength{\tabcolsep}{4pt}
\begin{tabular}{L{2.15cm}L{2.65cm}L{4.15cm}L{6.15cm}}
\toprule
\textbf{System} & \textbf{Primary setting} & \textbf{Main mechanism} & \textbf{Boundary relative to \RC} \\
\midrule
TrustChain \cite{malik2019trustchain} & IoT food supply chain & Sensor, buyer, and regulator evidence with time decay & Stronger product and multisource evidence context; weaker focus on repeated-pair discounting and newcomer fairness. \\
DeTRM \cite{putra2022detrm} & Blockchain supply chain & Commodity trust, agreement fulfilment, authority endorsement, and trust propagation & Stronger commodity and production-operation modelling; different from bounded participant updates under narrow interaction diversity. \\
MeritRank \cite{nasrulin2022meritrank} & DAO feedback graphs & Sybil-tolerant graph aggregation and decay & Stronger graph-level Sybil theory; not designed around permissioned contractual outcomes and supply-chain participation. \\
BlockRep \cite{zhao2025blockrep} & IIoT retail reviews & Purchasing proofs and cryptographic review authentication & Stronger review authenticity and anonymity; does not define dynamic participant trust under sparse history. \\
\RC & Permissioned supply-chain participation & Repeated-pair discount, diversity, identity confidence, and volume-aware decay & Targets bounded participant trust updates under concentrated evidence, identity uncertainty, and sparse history. \\
\bottomrule
\end{tabular}
\end{table*}

\section{System and Threat Model}
\subsection{Participants and Governed Interactions}
Let
\begin{equation}
E=\{E_1,E_2,\ldots,E_z\}
\end{equation}
denote the registered supply-chain participants. Each participant is bound to a blockchain-recognized identity with a public key, an organization, and an admission state. Trust-sensitive interactions form a directed graph $IG=(E,I)$, where an edge $I_{a,b}$ records a governed interaction from $E_a$ to $E_b$ under a contract or policy object.

Every interaction leaves a record: the request identifier, the participating identities, a contract reference, an outcome class, a timestamp, and an evidence digest. Outcome classes cover successful completion, delay, failure, default, and dispute. For every admitted entity, the framework maintains a reputation score $\R_i\in[0,1]$.

Figure~\ref{fig:architecture} shows how the pieces fit together. The blockchain plane stores admission records, interaction state, settlement or dispute evidence, and the update history. The reputation service reads committed evidence, computes the numerical update, and submits a bounded result together with its evidence and parameter references. A decision layer can then consume the score for credit, delivery, compliance, or any other risk-sensitive decision.

\begin{figure}[t]
\centering
\includegraphics[width=0.96\columnwidth]{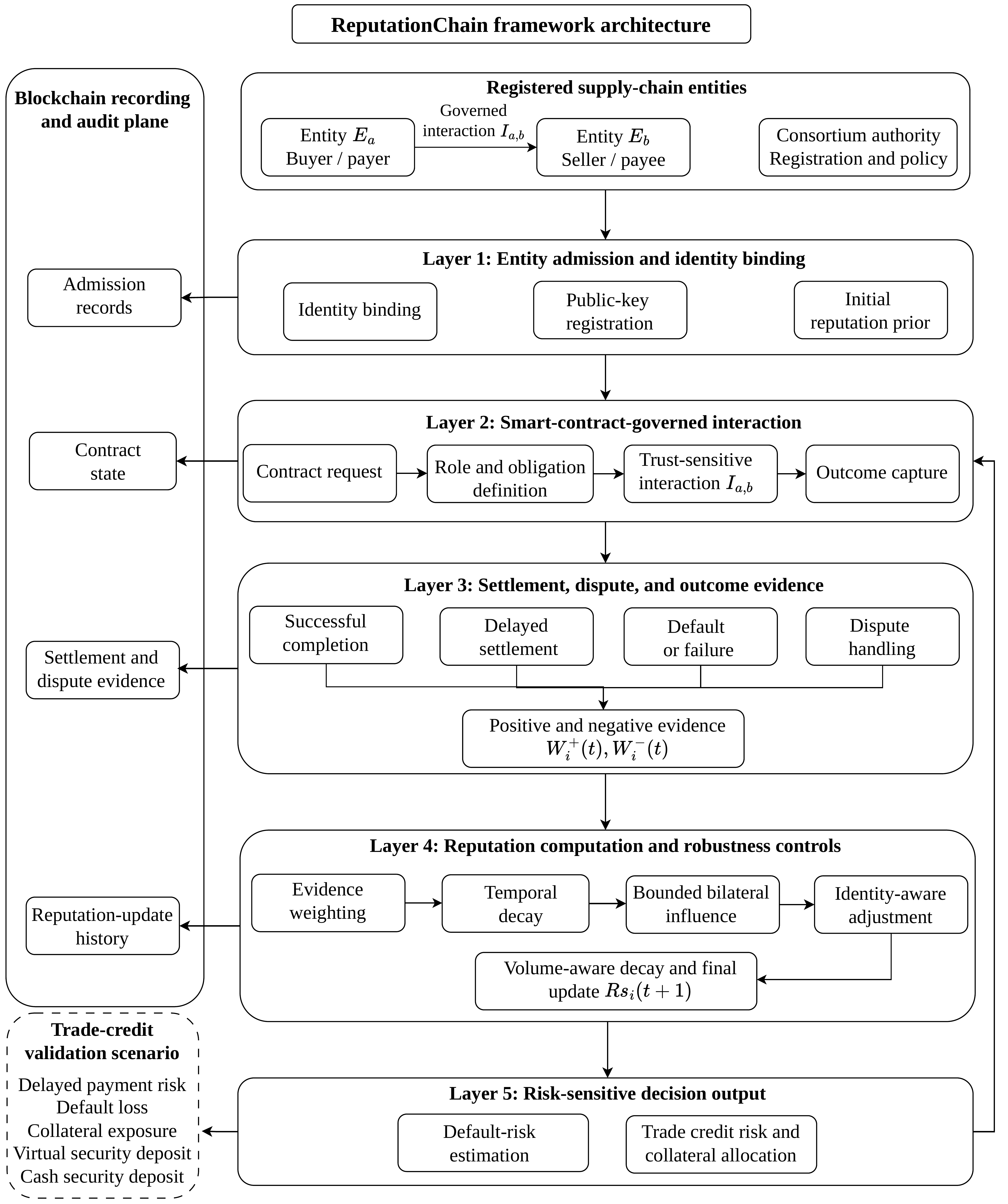}
\caption{\RC\ architecture. Blockchain preserves identity, contract, outcome, and update provenance. Numerical reputation computation is separated from the smart-contract path. The trade-credit block is an example of downstream risk-sensitive use, not a separate trust mechanism.}
\label{fig:architecture}
\end{figure}

\subsection{Adversary and Fairness Conditions}
The evaluation covers three bounded conditions.

Under \textit{repeated bilateral collusion}, a small set of entities keeps completing successful-looking interactions with the same counterparties, trying to pass correlated evidence off as independent.

Under \textit{identity multiplicity}, one controller operates $k$ registered identities. The system receives an identity-confidence value $\phi_i\in[0,1]$ for each identity as governance input, and \RC\ makes no attempt to infer $\phi_i$ from behavior. In an operational consortium that value would come from verifiable assurance evidence: legal-entity registration, the certificate assurance tier, independent organizational verification, or hardware attestation. Assurance that is set too high weakens the attenuation, which is exactly what the $\phi$ sweep in Section~\ref{sec:sensitivity} measures.

Under \textit{sparse-history fairness}, honest newcomers simply interact less often than established entities, and the risk is that time decay alone pushes them into false low-trust classification without any adverse behavior on their part.

Several conditions sit outside the evaluated scope: certificate-authority compromise, identity replacement, on-off attacks, strategic badmouthing, delayed reporting, and production network failures. We list them explicitly because, taken together, they define the current security boundary of the design.

\section{ReputationChain Model}
\subsection{Bounded Event Evidence}
For event $m$, let $Y_i^{(m)}\in\{0,1\}$ be the verified outcome for entity $E_i$, and let $r_m\in[0,1]$ be the interaction risk. With a neutral prior $R_0=0.5$, the bounded event evidence is
\begin{equation}
 e_i^{(m)}=\clip\left(\frac{1}{2}+r_m\left(Y_i^{(m)}-\frac{1}{2}\right)\right).
 \label{eq:eventevidence}
\end{equation}
A successful high-risk event pushes the evidence further above the prior, and a failed high-risk event pushes it further below.

One asymmetry runs through the whole design: when reinforcement looks suspicious, only positive evidence is weakened, and adverse evidence always lands at full strength. Without that asymmetry, colluders could turn the defense against itself and use the attenuation to soften the impact of their own failures.

\subsection{Repeated-Pair and Diversity Control}
Let $c_{i,j}^{(m)}$ count the observed interactions between $E_i$ and counterparty $E_j$ up to event $m$. The repeated-pair weight is
\begin{equation}
 \omega_{i,j}^{(m)}=\frac{1}{1+\eta(c_{i,j}^{(m)}-1)}, \qquad \eta\geq0,
 \label{eq:omega}
\end{equation}
which gives the first interaction full weight and every later interaction with the same pair progressively less.

Diversity enters through
\begin{equation}
D_i^{(m)}=\frac{|\mathcal{N}_i^{(m)}|}{N_i^{(m)}},
\label{eq:diversity}
\end{equation}
the ratio of distinct counterparties to total interactions; a low value signals concentrated evidence. In the collusion scenario, the raw attenuation is
\begin{equation}
q_i^{(m)}=\min\{1,\omega_{i,j_m}^{(m)}D_i^{(m)}\}.
\label{eq:qcoll}
\end{equation}

Figure~\ref{fig:collusionmechanism} contrasts this with the naive rule, which treats every repeated reciprocal event as if it were independent.

\begin{figure}[t]
\centering
\includegraphics[width=\columnwidth]{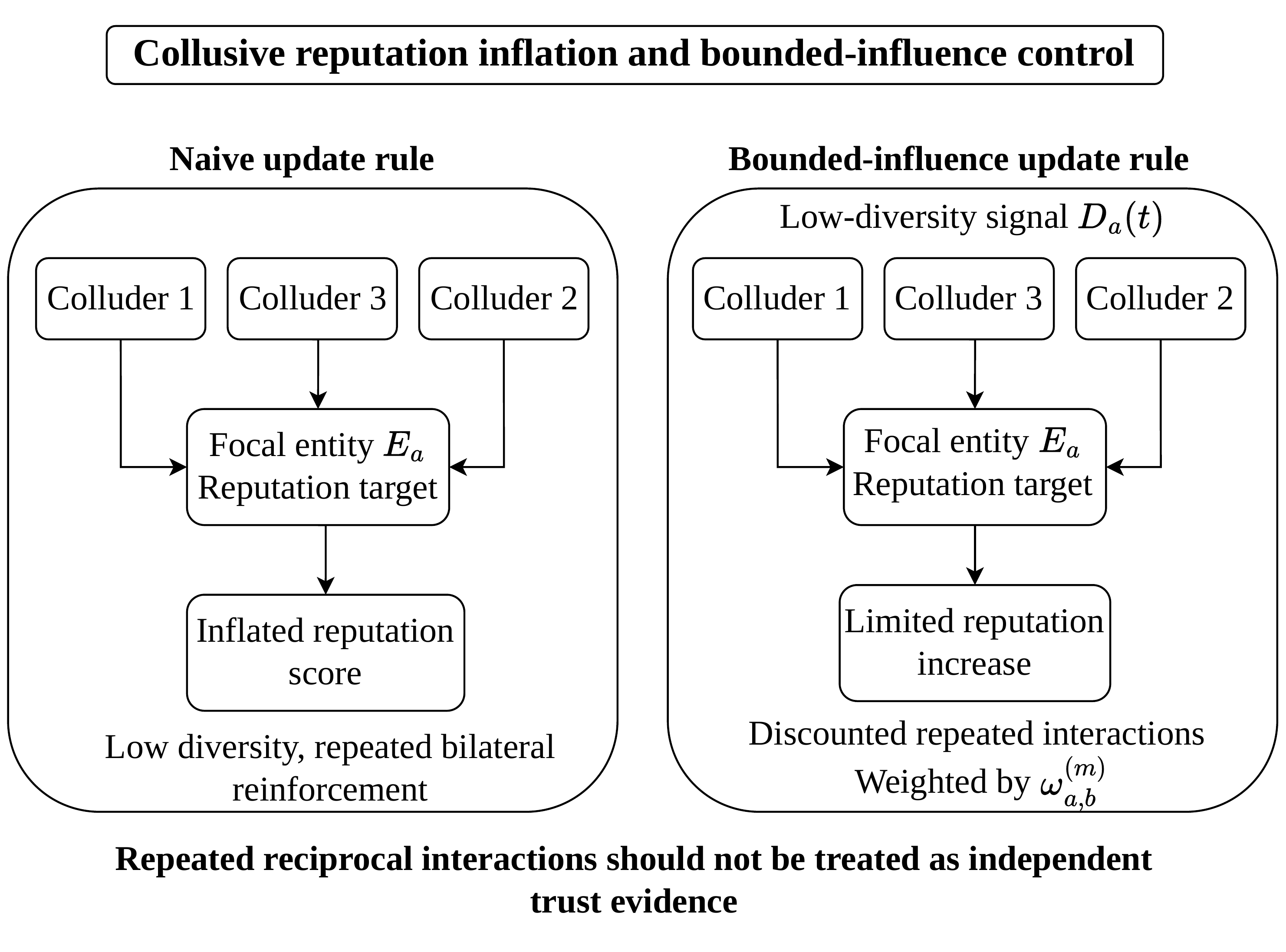}
\caption{Repeated reciprocal evidence can inflate reputation under a naive rule. \RC\ bounds this reinforcement through repeated-pair weight $\omega_{i,j}^{(m)}$ and diversity $D_i^{(m)}$.}
\label{fig:collusionmechanism}
\end{figure}

\subsection{Identity-Aware Adjustment}
For identity multiplicity, the attenuation becomes
\begin{equation}
q_i^{(m)}=\min\{1,\phi_iD_i^{(m)}\}.
\label{eq:qsybil}
\end{equation}
A lower $\phi_i$ attenuates positive evidence more strongly. Nothing here detects multiple identities; the mechanism is a controlled weighting applied after governance has already assigned an assurance level.

Both scenario-specific factors pass through a softening step,
\begin{equation}
s_i^{(m)}=1-\rho_{\mathrm{soft}}(1-q_i^{(m)}),
\quad 0\leq\rho_{\mathrm{soft}}\leq1,
\label{eq:soft}
\end{equation}
and the adjusted evidence is
\begin{equation}
\Delta_i^{(m)}=
\begin{cases}
R_0+s_i^{(m)}(e_i^{(m)}-R_0), & e_i^{(m)}>R_0,\\
e_i^{(m)}, & e_i^{(m)}\leq R_0.
\end{cases}
\label{eq:delta}
\end{equation}
The sequential update is
\begin{equation}
\R_i^{(m+1)}=\clip\left((1-\lambda)\R_i^{(m)}+\lambda\Delta_i^{(m)}\right),
\label{eq:update}
\end{equation}
where $0<\lambda\leq1$ sets how strongly the score responds to current evidence.

\subsection{Volume-Aware Decay and Newcomer Fairness}
A trust score should not sit frozen forever, but decay toward zero would treat inactivity as if it were failure. \RC\ instead decays toward the neutral prior $R_0$:
\begin{equation}
\R_i^{\mathrm{decay}}(t+\Delta t)=R_0+(\R_i(t)-R_0)
\exp\left(-\frac{\gamma\Delta t}{1+\kappa V_i(t)}\right),
\label{eq:decay}
\end{equation}
where $V_i(t)$ is the verified interaction volume, $\gamma>0$ is the base decay rate, and $\kappa\geq0$ controls the volume sensitivity. Setting $\kappa=0$ recovers static decay.

The interpretation needs some care. Higher volume slows the drift toward the neutral prior, which preserves established evidence; it hands nothing to newcomers for free. Newcomer fairness is therefore measured through the false low-trust rate, not through any unconditional protection of scores.

\subsection{Formal Properties}
\begin{proposition}[Boundedness]
If $\R_i^{(m)}\in[0,1]$, $\Delta_i^{(m)}\in[0,1]$, and $0<\lambda\leq1$, then the update in \eqref{eq:update} remains in $[0,1]$.
\end{proposition}
\begin{proof}[Proof sketch]
The un-clipped term is a convex combination of two bounded values. Clipping preserves the interval.
\end{proof}

\begin{proposition}[Sublinear repeated-pair influence]
For $N$ repeated interactions between the same pair and $\eta>0$,
\begin{equation}
\sum_{n=1}^{N}\frac{1}{1+\eta(n-1)}
\leq 1+\frac{1}{\eta}\ln(1+\eta(N-1)).
\end{equation}
Thus raw repeated-pair weight grows at most logarithmically rather than linearly.
\end{proposition}
\begin{proof}[Proof sketch]
Apply the integral bound for the decreasing function $1/x$ to the terms after the first.
\end{proof}

\begin{proposition}[Volume-aware attenuation]
For $\kappa\geq0$, volume-aware decay moves a score toward $R_0$ no faster than the corresponding static decay rule.
\end{proposition}
\begin{proof}[Proof sketch]
Since $\gamma/(1+\kappa V_i)\leq\gamma$, the volume-aware exponential factor is at least the static factor for $\Delta t\geq0$.
\end{proof}

These properties do not amount to attack detection. They establish boundedness and fix the direction of each mechanism; the experiments then measure how far those mechanisms actually move the outcomes.

\subsection{Blockchain Realization Boundary}
The blockchain stores admission, interaction, outcome, and update records. Every submitted update carries an evidence digest, a parameter-set version, a sequence number, and the score itself, and the smart contracts check authorization, evidence existence, contract state, update order, parameter version, and score bounds. The exponential calculation in \eqref{eq:decay} and the event sequence in \eqref{eq:update} run in an off-chain reputation service.

This separation buys two things and costs one. It keeps platform-specific floating-point behavior and repeated nonlinear computation out of the smart contracts. The cost is that an on-chain admissibility check cannot prove that a submitted value is the unique correct result. Any consortium member can recompute an update from the committed evidence and the referenced parameter version, so numerical correctness is auditable rather than cryptographically enforced. Work on trust as a service argues for the same split between evidence management and reputation computation \cite{olariu2024trustservice}.

\section{Evaluation}
\subsection{Questions, Metrics, and Protocol}
The evaluation asks three questions.
\begin{itemize}
  \item \textbf{Q1:} Does repeated-pair control reduce collusive reputation gain?
  \item \textbf{Q2:} Does identity-aware weighting reduce the advantage of identity multiplicity?
  \item \textbf{Q3:} Does volume-aware decay reduce false low-trust classification of honest newcomers?
\end{itemize}

Four metrics carry the results: collusive gain $G_{\mathrm{coll}}$, false trust elevation rate (FTER), reputation inflation ratio (RIR), and false low-trust rate (FLTR):
\begin{align}
G_{\mathrm{coll}} &= \overline{\R}_{\mathcal{C}}(T)-\overline{\R}_{\mathcal{C}}(0),\\
\mathrm{FTER} &= \frac{N_{\mathrm{malicious\ high\ trust}}}{N_{\mathrm{malicious}}},\\
\mathrm{RIR} &= \frac{\overline{\R}_{\mathrm{adversarial}}}{\overline{\R}_{\mathrm{honest}}},\\
\mathrm{FLTR} &= \frac{|\{i\in\mathcal{N}_{h}:\R_i(T)<\tau\}|}{|\mathcal{N}_{h}|}.
\end{align}
An RIR above one means the adversarial mean sits above the honest mean, and FLTR counts honest newcomers incorrectly left below the threshold $\tau$.

The simulator is written in Python and driven by fixed configuration files. Every reported condition aggregates 30 seeded runs, and within a scenario all methods consume the identical generated trace for each seed; this includes the static and volume-aware newcomer comparison, which shares traces in the same way. Sensitivity sweeps reuse the same seed identifiers across parameter values. Differences between methods therefore cannot be attributed to workload noise. Table~\ref{tab:settings} gives the populations, the outcome probabilities, and the model settings needed to reconstruct each scenario.

\begin{table*}[t]
\centering
\caption{Controlled simulation settings and workload-generation assumptions. Interaction risk $r_m$ is sampled uniformly from $[0.60,1.00]$ in every scenario.}
\label{tab:settings}
\scriptsize
\setlength{\tabcolsep}{4pt}
\begin{tabular}{L{1.9cm}L{4.2cm}L{5.2cm}L{4.2cm}}
\toprule
\textbf{Scenario} & \textbf{Population and workload} & \textbf{Outcome and interaction assumptions} & \textbf{Model settings} \\
\midrule
Global & 30 paired seeded runs & Initial and neutral reputation $R_0=0.50$; trust threshold $\tau=0.65$ & Update mixing $\lambda=0.35$ \\
Collusion & 100 entities; 10 colluders; 2,000 events over 60 days & Collusion-event probability 0.30; honest success 0.88; collusive success 0.98 & $\eta=0.80$; $\rho_{\mathrm{soft}}=0.60$ \\
Identity multiplicity & 80 honest entities; one controller; $k\in\{1,3,5,10\}$; 1,500 events & Base attack probability 0.25, raised by 0.03 per additional identity and capped at 0.80; honest success 0.88; controlled-identity positive outcome 0.99; intra-group fraction 0.60 & $\phi=0.35$; $\rho_{\mathrm{soft}}=0.55$ \\
Sparse history & 20 high-volume, 20 medium-volume, and 20 newcomer entities; 40 days & Daily Poisson interaction rates 4.0, 2.0, and 0.35; honest success 0.90 & $\gamma=0.10$; $\kappa=0.30$; at most 30 recovery events \\
Comparison models & Same trace per seed as the full model & B3 draws correlated synthetic sensor, buyer, and regulator ratings from the same event, with noise 0.05 & B2 decay 0.003 per event; B3 weights $(0.40,0.40,0.20)$ and decay 0.003 per event \\
\bottomrule
\end{tabular}
\end{table*}

\subsection{Baselines and Ablations}
Six labels are used across the scenarios.
\begin{itemize}
  \item B1 is naive mean evidence without sequential update, temporal decay, diversity, or identity weighting.
  \item B2 is a static-decay baseline.
  \item B3 is a TrustChain-inspired time-decayed multisource proxy \cite{malik2019trustchain}.
  \item B4 removes only the repetition-count discount. Diversity and softening remain active.
  \item B5 removes identity-confidence adjustment.
  \item B6 is full \RC.
\end{itemize}

B3 deserves a caveat. It is not a full TrustChain implementation, because independent real sensor, buyer, and regulator streams were not available; the proxy generates correlated synthetic components from the same underlying event. Its results should not be read as a general defeat of TrustChain. What it tests is whether time-decayed multisource aggregation alone holds up when the sources are correlated and manipulable.

\subsection{Seed-Paired Analysis}
Because every method sees the same trace for a given seed, the mechanism comparisons are naturally paired. The scenario tables report mean $\pm$ standard deviation for the raw metrics; Table~\ref{tab:paired} reports the mean paired effect with a 95\% Student-$t$ confidence interval, together with a one-sided Wilcoxon signed-rank test as a distribution-free check on the expected direction.\footnote{Paired effects are computed from unrounded run-level values, so they may differ in the last digit from the difference of the rounded means shown in the scenario tables.} Readers should note that the paired intervals are far tighter than the raw standard deviations would suggest, which is the expected consequence of common random numbers: the seed-to-seed workload variation cancels inside each difference. Four of the five effects are consistent across all 30 matched runs. The recovery-burden difference is the exception, and its interval includes zero.

\begin{table}[t]
\centering
\caption{Mean paired effects over 30 matched seeds. Positive values favor the full or volume-aware model.}
\label{tab:paired}
\scriptsize
\resizebox{\columnwidth}{!}{%
\begin{tabular}{lcc}
\toprule
\textbf{Comparison and metric} & \textbf{Mean effect [95\% CI]} & \textbf{Wilcoxon} \\
\midrule
B4 $-$ B6 collusive gain & $0.0372~[0.0361,0.0383]$ & $p<0.001$ \\
B5 $-$ B6 RIR at $k=10$ & $0.0506~[0.0496,0.0515]$ & $p<0.001$ \\
Volume-aware $-$ static newcomer rep. & $0.0963~[0.0929,0.0997]$ & $p<0.001$ \\
Static $-$ volume-aware FLTR & $0.1950~[0.1613,0.2287]$ & $p<0.001$ \\
Static $-$ volume-aware recovery & $0.0183~[-0.1186,0.1552]$ & $p=0.168$ \\
\bottomrule
\end{tabular}%
}
\end{table}

\subsection{Repeated Bilateral Collusion}
Table~\ref{tab:collusion} reports the collusion results. B6 produces the lowest collusive gain at 0.1443, B4 reaches 0.1815, and the simple baselines all stay above 0.35. Because B4 keeps diversity and softening active, the gap between B4 and B6 isolates what the repeated-pair discount adds on its own.

\begin{table}[t]
\centering
\caption{Repeated bilateral collusion, mean $\pm$ standard deviation.}
\label{tab:collusion}
\scriptsize
\resizebox{\columnwidth}{!}{%
\setlength{\tabcolsep}{3.2pt}
\begin{tabular}{lccc}
\toprule
\textbf{Model} & \textbf{Gain} & \textbf{FTER} & \textbf{Colluder rep.} \\
\midrule
B1 & $0.3688\pm0.0055$ & $1.0000\pm0$ & $0.8688\pm0.0055$ \\
B2 & $0.3585\pm0.0277$ & $0.9467\pm0.0776$ & $0.8585\pm0.0277$ \\
B3 & $0.3806\pm0.0097$ & $1.0000\pm0$ & $0.8806\pm0.0097$ \\
B4 & $0.1815\pm0.0207$ & $0.8533\pm0.1167$ & $0.6815\pm0.0207$ \\
\textbf{B6} & $\mathbf{0.1443\pm0.0197}$ & $0.7333\pm0.1539$ & $\mathbf{0.6443\pm0.0197}$ \\
\bottomrule
\end{tabular}%
}
\end{table}

\begin{figure}[t]
\centering
\includegraphics[width=\columnwidth]{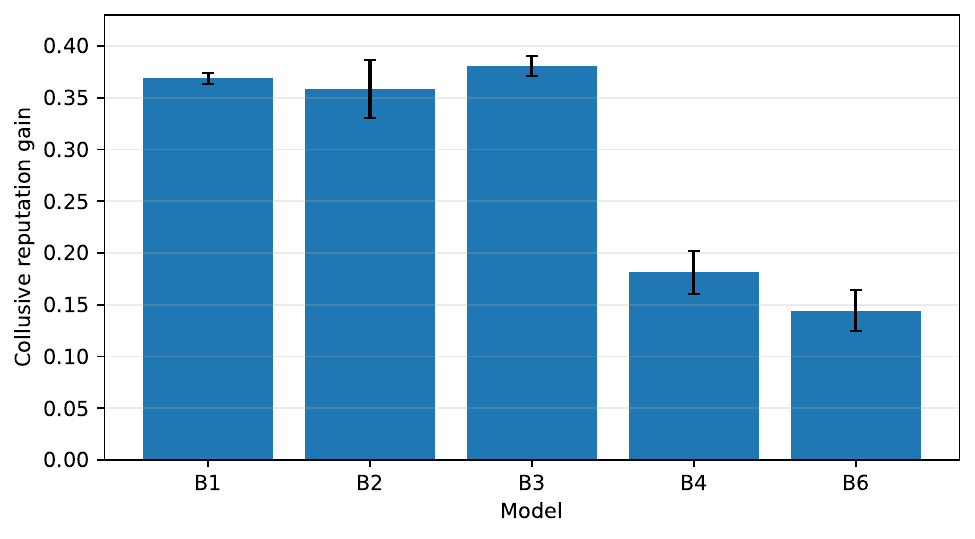}
\caption{Collusive reputation gain. Error bars show standard deviation over 30 seeded runs.}
\label{fig:collusionresults}
\end{figure}

The residual FTER of 0.7333 deserves emphasis. Colluders still complete successful-looking events, and many of them stay above the trust threshold, so nothing in this table amounts to colluder classification. The claim the data supports is smaller: repeated-pair and diversity control shrink the magnitude of the inflation.

\subsection{Identity Multiplicity}
At $k=10$, B6 reaches an RIR of 0.8723, while B1, B2, and B3 remain above 1.08 and B5 lands at 0.9228. The B5 to B6 difference isolates the identity-confidence weighting. Figure~\ref{fig:sybilresults} traces the trend over all tested values of $k$.

\begin{table}[t]
\centering
\caption{Identity multiplicity at $k=10$, mean $\pm$ standard deviation.}
\label{tab:sybil}
\scriptsize
\resizebox{\columnwidth}{!}{%
\setlength{\tabcolsep}{3.2pt}
\begin{tabular}{lccc}
\toprule
\textbf{Model} & \textbf{RIR} & \textbf{FTER} & \textbf{Identity rep.} \\
\midrule
B1 & $1.0838\pm0.0124$ & $1.0000\pm0$ & $0.8918\pm0.0042$ \\
B2 & $1.0908\pm0.0277$ & $0.9933\pm0.0254$ & $0.8935\pm0.0171$ \\
B3 & $1.1047\pm0.0164$ & $1.0000\pm0$ & $0.9077\pm0.0049$ \\
B5 & $0.9228\pm0.0218$ & $0.9833\pm0.0531$ & $0.7387\pm0.0123$ \\
\textbf{B6} & $\mathbf{0.8723\pm0.0207}$ & $0.9800\pm0.0551$ & $\mathbf{0.6983\pm0.0111}$ \\
\bottomrule
\end{tabular}%
}
\end{table}

\begin{figure}[t]
\centering
\includegraphics[width=\columnwidth]{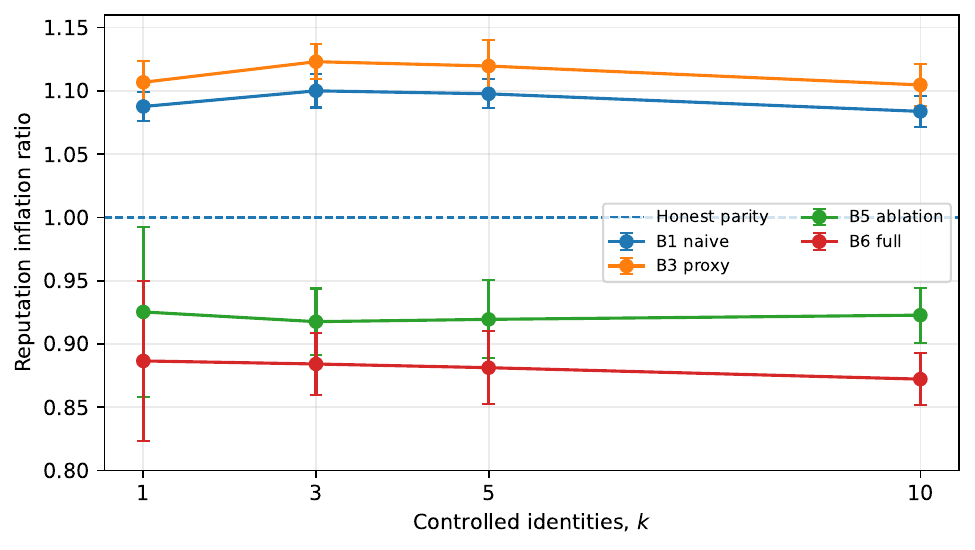}
\caption{RIR across controlled identity multiplicity. Values below one remain below the honest mean. B3 is a controlled TrustChain-inspired proxy, not a full reimplementation.}
\label{fig:sybilresults}
\end{figure}

FTER stays near 0.98 even for B6, so the experiment establishes no Sybil detection at all. What it does establish is reduced reputation amplification once low identity assurance is already known, which makes the mechanism a policy instrument rather than a detector.

\subsection{Newcomer Fairness}
Table~\ref{tab:newcomer} compares static and volume-aware decay on identical newcomer traces. Volume-aware decay lifts mean newcomer reputation from 0.6626 to 0.7589 and cuts FLTR from 0.3633 to 0.1683, while the conditional recovery burden barely moves (from 1.2772 to 1.2590, with a paired interval that includes zero). The benefit is preventive, in other words: fewer honest newcomers get pushed below the threshold in the first place, and recovery after a low-trust classification becomes no easier than before.

\begin{table}[t]
\centering
\caption{Honest newcomer fairness, mean $\pm$ standard deviation.}
\label{tab:newcomer}
\scriptsize
\resizebox{\columnwidth}{!}{%
\setlength{\tabcolsep}{3pt}
\begin{tabular}{lccc}
\toprule
\textbf{Decay} & \textbf{Mean rep.} & \textbf{FLTR} & \textbf{Recovery} \\
\midrule
Static & $0.6626\pm0.0168$ & $0.3633\pm0.0909$ & $1.2772\pm0.1529$ \\
Volume-aware & $\mathbf{0.7589\pm0.0221}$ & $\mathbf{0.1683\pm0.0895}$ & $1.2590\pm0.3933$ \\
\bottomrule
\end{tabular}%
}
\end{table}

\begin{figure}[t]
\centering
\includegraphics[width=\columnwidth]{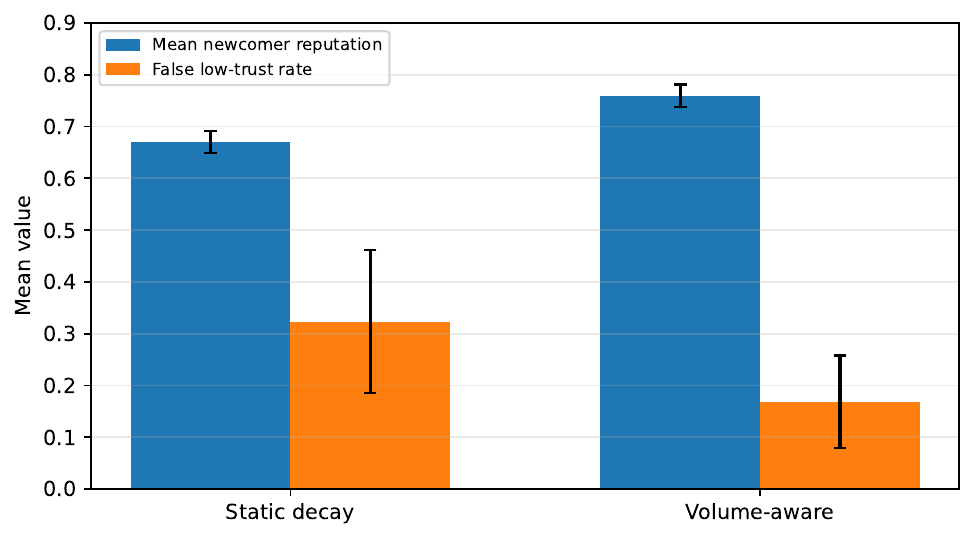}
\caption{Mean newcomer reputation and false low-trust rate. Error bars show standard deviation over 30 runs.}
\label{fig:newcomerresults}
\end{figure}

\subsection{Sensitivity and Trade-Offs}
\label{sec:sensitivity}
The one-factor-at-a-time sweeps vary $\eta$, $\kappa$, and $\phi$ under common seeds. Table~\ref{tab:sensitivity} lists selected endpoints beside the reference configuration. Three of the boundary values double as a consistency check on the implementation, because each one reduces the full model to a configuration we also ran as a separate code path: $\eta=0$ is B4, $\phi=1$ is B5, and $\kappa=0$ is static decay. The sweep endpoints reproduce those ablation results exactly.

\begin{table}[t]
\centering
\caption{Sensitivity summary. The middle value is the main reference setting.}
\label{tab:sensitivity}
\scriptsize
\resizebox{\columnwidth}{!}{%
\begin{tabular}{lccc}
\toprule
\textbf{Sweep} & \textbf{Low or boundary} & \textbf{Reference} & \textbf{High} \\
\midrule
$\eta$: collusive gain & 0.1815 at 0.00 & \textbf{0.1443 at 0.80} & 0.1426 at 1.40 \\
$\eta$: honest rep. & 0.7817 & \textbf{0.7623} & 0.7566 \\
$\kappa$: FLTR & 0.3633 at 0.00 & \textbf{0.1683 at 0.30} & 0.1533 at 0.80 \\
$\kappa$: newcomer rep. & 0.6626 & \textbf{0.7589} & 0.7852 \\
$\phi$: RIR & 0.8606 at 0.20 & \textbf{0.8723 at 0.35} & 0.9228 at 1.00 \\
$\phi$: identity rep. & 0.6889 & \textbf{0.6983} & 0.7387 \\
\bottomrule
\end{tabular}%
}
\end{table}

Pushing $\eta$ past 0.80 buys only small further reductions in collusive gain and pays for them with lower honest reputation, a plain robustness-utility trade-off. Raising $\kappa$ improves newcomer outcomes with diminishing returns, and raising $\phi$ weakens the attenuation and lets RIR climb, exactly as the model predicts. None of the sweeps points to a universal optimum, so calibration remains an operational task rather than a solved problem.

\section{Discussion and Limitations}
\subsection{What the Results Establish}
The strongest evidence comes from the matched baselines and ablations. B4 against B6 isolates the repeated-pair discount, B5 against B6 isolates the identity-confidence weighting, and static against volume-aware decay isolates the sparse-history mechanism. Common traces strip out workload noise, the paired intervals show that each effect holds run by run rather than only on average, and the sensitivity sweeps show that the directional findings do not hinge on one exact parameter value.

Three bounded statements survive this scrutiny. Repeated bilateral evidence receives diminishing influence and yields lower collusive gain. Governance-supplied identity confidence reduces the advantage of multiplicity. Volume-aware decay lowers false low-trust classification under the tested sparse-history workload.

\subsection{What the Results Do Not Establish}
The evaluation is a controlled simulation, not a deployment benchmark: the interaction outcomes and adversarial behavior are synthetic, the parameters are transparent but never calibrated against operational supply-chain records, and the B3 comparison is a proxy. Real seasonal demand, dispute cycles, delayed reporting, and correlated organizational behavior could all shift the absolute reputation trajectories.

The attack model is also incomplete. Whitewashing, certificate replacement, on-off attacks, badmouthing, strategic timing, and governance compromise all remain open, and the model assumes that adverse outcomes are captured correctly in the first place. A blockchain protects committed evidence from later alteration; it cannot certify that an external event was honestly observed before submission.

Finally, the off-chain computation boundary leaves an integrity dependency. Consortium members can recompute every update, but the smart contract never proves the nonlinear calculation itself. Deterministic fixed-point arithmetic, replicated computation, threshold approval of updates, or succinct proofs could each strengthen this in future designs.

\subsection{Future Work}
The natural next step is a working permissioned-blockchain prototype with measured transaction latency, update cost, and storage overhead, driven by operational or realistically derived interaction traces and used to calibrate $\eta$, $\kappa$, $\phi$, $\lambda$, and $\tau$. Adaptive adversaries belong in that evaluation as well, including on-off behavior, whitewashing, badmouthing, identity replacement, and governance that assigns assurance incorrectly. A second application beyond trade credit would then test whether the mechanism generalizes past a single risk scenario. Each of these additions probes something the present study cannot: external validity, implementation cost, and attack coverage.

\section{Conclusion}
We presented \RC, a blockchain-supported participant trust framework for governed supply-chain interactions. The model bounds repeated bilateral reinforcement, folds in governance-supplied identity confidence, and decays scores toward a neutral prior at a volume-aware rate. Across controlled simulations with 30 seeded runs and matched traces, it lowers collusive gain, identity-multiplicity inflation, and false low-trust classification of honest newcomers, and the seed-paired analysis shows those effects are consistent rather than averaged artifacts. The evidence supports a bounded reduction of reputation distortion rather than attack detection, and production use still calls for deployment measurement, calibration on real data, stronger numerical verification, and a broader adversarial evaluation.

\section*{Data Availability}
Every value reported here is drawn from a frozen set of audited raw and summary tables, and the result figures are regenerated from those tables by a single plotting script. All five paired effects in Table~\ref{tab:paired} reproduce exactly from the retained per-seed records. The simulator, its configuration files, the raw and summary tables, and the plotting script are available at github.

\section*{Declaration on the Use of AI Tools}
The AI was used for language editing and for restructuring the prose of this manuscript, and for assembling the \LaTeX{} source. The research design, the simulation code, the experimental configurations, the data, and the figures are the authors' own work and were neither generated nor selected by an AI system. The authors have verified every reported value against the underlying result files and accept full responsibility for the content.
%

\FloatBarrier
\balance
\begingroup
\footnotesize
\bibliographystyle{IEEEtran}
\bibliography{references}
\endgroup

\end{document}